\documentclass{article}

\usepackage[utf8]{inputenc}
\usepackage{import}

\ifdefined\uselmodern
\usepackage{lmodern}
\fi

\usepackage[T1]{fontenc}
\usepackage{etex}
\usepackage{color}
\usepackage{microtype}
\usepackage{etoolbox}\usepackage{ifthen}
\usepackage{xparse}
\usepackage{substr}
\usepackage{subcaption}
\usepackage[hide=true,dir=bashful]{bashful}
\usepackage{calculator}
\usepackage{xspace}

\makeatletter
\@ifclassloaded{imaiai}{}{
  \@ifclassloaded{IEEEtran}{
      }{
    \@ifundefined{nobib}{
      \usepackage{biblatex}
      \addbibresource{IEEEabrv.bib}
      \addbibresource{literature.bib}
    }
  }
}
\makeatother

\usepackage{mathtools}

\usepackage{booktabs}
\usepackage{array}

\usepackage{dsfont}
\usepackage{mathrsfs}

\usepackage{setspace}

\usepackage[justification=centering]{caption}

\usepackage{colonequals}
\usepackage{centernot}
\usepackage{relsize}
\usepackage{tikz}
\usepackage{pgfplots}
\usepackage{longtable}
\usepackage{eqparbox}
\usepackage{fp}

\usepackage{comment}

\makeatletter
\@ifclassloaded{beamer}{
  \usepackage{appendixnumberbeamer}
}{
  \usepackage{amssymb}
  \usepackage{environ}
  \usepackage{amsmath}
  \usepackage[bottom]{footmisc}
  \usepackage[thepage]{zref}
  \usepackage{enumitem}
  \@ifundefined{nolinks}{
    \usepackage{hyperref}
    \@ifclassloaded{a0poster}{}{
      \usepackage[all]{hypcap}
    }
    \usepackage{bookmark}
  }{
    
  }
}
\makeatother

\usepackage{amsthm}
\usepackage{thmtools, thm-restate}
\usepackage[capitalise]{cleveref}

\ifdefined\useautonum
\makeatletter
\@ifclassloaded{beamer}{
  \usepackage{scrpage2}
}{
  \usepackage{scrlayer-scrpage}
  \usepackage{autonum}
  \ifdef{\Cref}{%
    \autonum@generatePatchedReferenceCSL{Cref}%
  }{}%
}
\makeatother
\fi

\usepackage{scalerel}
\usepackage{stackengine,wasysym}

\usepackage{fixmath}

\crefname{apx}{appendix}{appendices}
\Crefname{apx}{Appendix}{Appendices}
\crefformat{apx}{the #2appendix#3}
\Crefformat{apx}{the #2Appendix#3}

\newtheorem*{conjecture*}{Conjecture}
\newtheorem*{definition*}{Definition}
\newtheorem*{proposition*}{Proposition}
\newtheorem*{corollary*}{Corollary}
\makeatletter
\@ifclassloaded{beamer}{
  
  \newtheorem*{lemma*}{Lemma}
  \newtheorem*{theorem*}{Theorem}
  \newtheorem*{example*}{Example}
}{
  \newcommand{\mynewtheorem}[2]{
    \ifdef{\seperatenumbering}{
      \newtheorem{#1}{#2}
    }{
      \newtheorem{#1}[theorem]{#2}
    }
  }
  \@ifclassloaded{imaiai}{}{
    \ifdef{\numbertheoremwithin}{
      \newtheorem{theorem}{Theorem}[\numbertheoremwithin]
    }{
      \newtheorem{theorem}{Theorem}
    }
    \newtheorem*{theorem*}{Theorem}
    
    \mynewtheorem{lemma}{Lemma}
    \mynewtheorem{corollary}{Corollary}
    \mynewtheorem{definition}{Definition}
    \mynewtheorem{proposition}{Proposition}
    \mynewtheorem{conjecture}{Conjecture}

    \theoremstyle{remark}
    
    \newtheorem{remark}{Remark}
    \newtheorem*{remark*}{Remark}
    \newtheorem*{example*}{Example}
  }
}
\makeatother

\crefname{conjecture}{Conjecture}{Conjectures}
\Crefname{conjecture}{Conjecture}{Conjectures}

\crefname{enumi}{part}{parts}
\Crefname{enumi}{Part}{Parts}

\crefname{equation}{Equation}{Equations}
\Crefname{equation}{Equation}{Equations}
\crefformat{equation}{#2(#1)#3}
\crefrangeformat{equation}{#3(#1)#4--#5(#2)#6}
\crefmultiformat{equation}{#2(#1)#3}{ and~#2(#1)#3}{, #2(#1)#3}{ and~#2(#1)#3}
\Crefformat{equation}{#2(#1)#3}
\Crefrangeformat{equation}{#3(#1)#4--#5(#2)#6}
\Crefmultiformat{equation}{#2(#1)#3}{ and~#2(#1)#3}{, #2(#1)#3}{ and~#2(#1)#3}
\crefname{figure}{Figure}{Figures}
\Crefname{figure}{Figure}{Figures}

\crefname{page}{page}{page}
\Crefname{page}{Page}{Page}

\makeatletter
\@ifclassloaded{IEEEtran}{}{
  \crefname{figure}{Figure}{Figures}
  \Crefname{figure}{Figure}{Figures}
}
\makeatother

\crefname{assumption}{assumption}{assumptions}
\Crefname{assumption}{Assumption}{Assumptions}

\crefname{case}{case}{cases}
\Crefname{case}{Case}{Cases}

\let\originalleft\left
\let\originalright\right
\let\originalmiddle\middle
\renewcommand{\left}{\mathopen{}\mathclose\bgroup\originalleft}
\renewcommand{\right}{\aftergroup\egroup\originalright}
\renewcommand{\middle}{\originalmiddle}

\makeatletter
\@ifclassloaded{beamer}{%
  \newcommand{\todo}[1][\empty]{%
    \textbf{\textcolor{red}{\ifthenelse{\equal{#1}{\empty}}{TODO}{TODO:~\emph{#1}}}}%
  }%
}
{
  \newcounter{todo}
  \newcommand{\todo}[1][\empty]{%
    \bookmarksetupnext{keeplevel}    \subpdfbookmark{TODO}{todo.\arabic{todo}}%
    \bookmarksetup{keeplevel=false}%
    \refstepcounter{todo}%
    \textcolor{red}{\textbf{TODO}%
      \ifthenelse{\equal{#1}{\empty}}{}%
      {\footnote{\textcolor{red}{\emph{#1}}}}%
    }%
    \xspace%
  }
}
\makeatother

\makeatletter
\@ifclassloaded{beamer}{
  \newcommand\changemode[1]{%
    \gdef\beamer@currentmode{#1}}
}{}
\makeatother

\makeatletter
\@ifclassloaded{IEEEtran}{
  
}{}
\makeatother

\newcounter{align}
\expandafter\let\expandafter\oldalign\csname align\endcsname
\expandafter\def\csname align\endcsname{\refstepcounter{align}\oldalign}

\newcommand{\ie}{i.\,e.\@\xspace}

\newcommand{\eg}{e.\,g.\@\xspace}

\newcommand{\cf}{cf.\@\xspace}

\newcommand{\iid}{i.i.d.\@\xspace}

\newcommand{\wrt}{w.r.t.\@\xspace}

\newcommand{\eqnl}{\nonumber\\*}

\DeclareDocumentCommand{\card}{ O{\empty} m}{
  \ensuremath{
    \ifthenelse{\equal{#1}{big}}{\big\vert #2 \big\vert}{
      \ifthenelse{\equal{#1}{Big}}{\Big\vert #2 \Big\vert}{
        \ifthenelse{\equal{#1}{bigg}}{\bigg\vert #2 \bigg\vert}{
          \ifthenelse{\equal{#1}{Bigg}}{\Bigg\vert #2 \Bigg\vert}{
            \ifthenelse{\equal{#1}{normal}}{\vert #2 \vert}{\left\vert #2 \right\vert}
          }
        }
      }
    }
  }
}

\DeclareDocumentCommand{\Vol}{ O{\empty} m}{
  \ensuremath{
    \ensuremath{\mathrm{Vol}
      \ifthenelse{\equal{#1}{big}}{\big(#2\big)}{
        \ifthenelse{\equal{#1}{Big}}{\Big(#2\Big)}{
          \ifthenelse{\equal{#1}{bigg}}{\bigg(#2\bigg)}{
            \ifthenelse{\equal{#1}{Bigg}}{\Bigg(#2\Bigg)}{
              \ifthenelse{\equal{#1}{normal}}{(#2)}{\left(#2\right)}
            }
          }
        }
      }
    }
  }
}

\DeclareDocumentCommand{\cvx}{ O{\empty} m}{
  \ensuremath{
    \ensuremath{\mathrm{conv}
      \ifthenelse{\equal{#1}{big}}{\big(#2\big)}{
        \ifthenelse{\equal{#1}{Big}}{\Big(#2\Big)}{
          \ifthenelse{\equal{#1}{bigg}}{\bigg(#2\bigg)}{
            \ifthenelse{\equal{#1}{Bigg}}{\Bigg(#2\Bigg)}{
              \ifthenelse{\equal{#1}{normal}}{(#2)}{\left(#2\right)}
            }
          }
        }
      }
    }
  }
}

\newcommand{\norm}[1]{\ensuremath{\left\Vert{#1}\right\Vert}}

\newcommand{\ol}[1]{\ensuremath{\overline{#1}}}

\DeclareDocumentCommand{\Ntoo}{ O{\empty} m}{
  \ensuremath{
    \ifthenelse{\equal{#1}{big}}{\big[#2\big]}{
      \ifthenelse{\equal{#1}{Big}}{\Big[#2\Big]}{
        \ifthenelse{\equal{#1}{bigg}}{\bigg[#2\bigg]}{
          \ifthenelse{\equal{#1}{Bigg}}{\Bigg[#2\Bigg]}{
            \ifthenelse{\equal{#1}{normal}}{[#2]}{\left[#2\right]}
          }
        }
      }
    }
  }
}
\DeclareDocumentCommand{\NZtoo}{ O{\empty} m}{
  \ensuremath{
    \ifthenelse{\equal{#1}{big}}{\big[0\,{:}\,#2\big]}{
      \ifthenelse{\equal{#1}{Big}}{\Big[0\,{:}\,#2\Big]}{
        \ifthenelse{\equal{#1}{bigg}}{\bigg[0\,{:}\,#2\bigg]}{
          \ifthenelse{\equal{#1}{Bigg}}{\Bigg[0\,{:}\,#2\Bigg]}{
            \ifthenelse{\equal{#1}{normal}}{[0\,{:}\,#2]}{\left[0\,{:}\,#2\right]}
          }
        }
      }
    }
  }
}

\newcommand{\Nto}[1]{\ensuremath{\{1,2,\dots,#1\}}}

\newcommand{\defas}{\vcentcolon=}

\newcommand{\ee}[1]{\ensuremath{\mathrm{e}^{#1}}}

\newcommand{\wt}[1]{\ensuremath{\widetilde{#1}}}

\newcommand{\compl}{\ensuremath{^{\mathrm{c}}}}

\newcommand{\NN}{\ensuremath{\mathbb{N}}}

\newcommand{\RR}{\ensuremath{\mathbb{R}}}

\newcommand{\dd}{\ensuremath{\mathrm{d}}}

\newcommand{\vt}[1]{\ensuremath{\boldsymbol{#1}}}

\newcommand{\wc}{\ensuremath{{}\cdot{}}}

\DeclareDocumentCommand{\ind}{m m}{%
  \ensuremath{%
    \mathds{1}_{#1}%
    \ifthenelse{\equal{#2}{\empty}}{}{(#2)}%
  }%
}

\DeclareDocumentCommand{\Exp}{ O{\empty} O{\empty} m}{
  \ensuremath{\mathds{E}_{#1}
    \ifthenelse{\equal{#2}{big}}{\big[ #3 \big]}{
      \ifthenelse{\equal{#2}{Big}}{\Big[ #3 \Big]}{
        \ifthenelse{\equal{#2}{bigg}}{\bigg[ #3 \bigg]}{
          \ifthenelse{\equal{#2}{Bigg}}{\Bigg[ #3 \Bigg]}{
            \ifthenelse{\equal{#2}{normal}}{[ #3 ]}{\left[ #3 \right]}
          }
        }
      }
    }
  }
}

\DeclareDocumentCommand{\condExp}{ O{\empty} m m}{
  \ensuremath{\mathds{E}
    \ifthenelse{\equal{#1}{big}}{\big[{#2\big| #3}\big]}{
      \ifthenelse{\equal{#1}{Big}}{\Big[{#2\Big| #3}\Big]}{
        \ifthenelse{\equal{#1}{bigg}}{\bigg[{#2\bigg| #3}\bigg]}{
          \ifthenelse{\equal{#1}{Bigg}}{\Bigg[{#2\Bigg| #3}\Bigg]}{
            \ifthenelse{\equal{#1}{normal}}{[{#2| #3}]}{\left[#2 \middle| #3\right]}
          }
        }
      }
    }
  }
}

\newcommand{\p}[2][\empty]{\ensuremath{\mathrm{p}_{#1}\ifthenelse{\equal{#2}{}}{}{\left({#2}\right)}}}
\newcommand{\op}[2][\empty]{\ensuremath{\ol{\mathrm{p}}_{#1}\ifthenelse{\equal{#2}{}}{}{\left({#2}\right)}}}
\newcommand{\q}[2][\empty]{\ensuremath{\mathrm{q}_{#1}\ifthenelse{\equal{#2}{}}{}{\left({#2}\right)}}}
\newcommand{\wtp}[2][\empty]{\ensuremath{\wt{\mathrm{p}}_{#1}\ifthenelse{\equal{#2}{}}{}{\left({#2}\right)}}}
\newcommand{\htp}[2][\empty]{\ensuremath{\hat{\mathrm{p}}_{#1}\ifthenelse{\equal{#2}{}}{}{\left({#2}\right)}}}
\newcommand{\pp}[2][\empty]{\ensuremath{\mathrm{p}'_{#1}\ifthenelse{\equal{#2}{}}{}{\left({#2}\right)}}}
\newcommand{\qq}[2][\empty]{\ensuremath{\mathrm{q}'_{#1}\ifthenelse{\equal{#2}{}}{}{\left({#2}\right)}}}
\newcommand{\ppp}[2][\empty]{\ensuremath{\mathrm{p}''_{#1}\ifthenelse{\equal{#2}{}}{}{\left({#2}\right)}}}

\newcommand{\pdf}[2][\empty]{\ensuremath{\mathrm{f}_{#1}\ifthenelse{\equal{#2}{}}{}{\left({#2}\right)}}}

\DeclareDocumentCommand{\Prob}{ O{\empty} O{\empty} m}{
  \ensuremath{\mathrm{P}\ifthenelse{\equal{#1}{\empty}}{}{_{#1}}
    \ifthenelse{\equal{#2}{big}}{\big\{#3\big\}}{
      \ifthenelse{\equal{#2}{Big}}{\Big\{#3\Big\}}{
        \ifthenelse{\equal{#2}{bigg}}{\bigg\{#3\bigg\}}{
          \ifthenelse{\equal{#2}{Bigg}}{\Bigg\{#3\Bigg\}}{
            \ifthenelse{\equal{#2}{normal}}{\{#3\}}{\left\{#3\right\}}
          }
        }
      }
    }
  }
}

\DeclareDocumentCommand{\DKL}{ O{\empty} m m}{
  \ensuremath{\mathrm{D}
    \ifthenelse{\equal{#1}{big}}{\big(#2 \big\Vert #3\big)}{
      \ifthenelse{\equal{#1}{Big}}{\Big(#2 \Big\Vert #3\Big)}{
        \ifthenelse{\equal{#1}{bigg}}{\bigg(#2 \bigg\Vert #3\bigg)}{
          \ifthenelse{\equal{#1}{Bigg}}{\Bigg(#2 \Bigg\Vert #3\Bigg)}{
            \ifthenelse{\equal{#1}{normal}}{(#2 \Vert #3)}{\left(#2 \middle\Vert #3\right)}
          }
        }
      }
    }
  }
}

\DeclareDocumentCommand{\pcond}{ O{\empty} O{\empty} O{\empty} m m}{
  \ensuremath{\mathrm{p}
    \ifthenelse{\equal{#1#2}{\empty}}{}{_{#1|#2}}
    \ifthenelse{\equal{#4#5}{\empty}}{}{
      \ifthenelse{\equal{#3}{big}}{\big(#4 \big| #5\big)}{
        \ifthenelse{\equal{#3}{Big}}{\Big(#4 \Big| #5\Big)}{
          \ifthenelse{\equal{#3}{bigg}}{\bigg(#4 \bigg| #5\bigg)}{
            \ifthenelse{\equal{#3}{Bigg}}{\Bigg(#4 \Bigg| #5\Bigg)}{
              \ifthenelse{\equal{#3}{normal}}{(#4 | #5)}{\left(#4 \middle| #5\right)}
            }
          }
        }
      }
    }
  }
}

\DeclareDocumentCommand{\Pcond}{ O{\empty} O{\empty} O{\empty} m m}{
  \ensuremath{\mathrm{P}
    \ifthenelse{\equal{#1#2}{\empty}}{}{_{#1|#2}}
    \ifthenelse{\equal{#4#5}{\empty}}{}{
      \ifthenelse{\equal{#3}{big}}{\big\{#4 \big| #5\big\}}{
        \ifthenelse{\equal{#3}{Big}}{\Big\{#4 \Big| #5\Big\}}{
          \ifthenelse{\equal{#3}{bigg}}{\bigg\{#4 \bigg| #5\bigg\}}{
            \ifthenelse{\equal{#3}{Bigg}}{\Bigg\{#4 \Bigg| #5\Bigg\}}{
              \ifthenelse{\equal{#3}{normal}}{\{#4 | #5\}}{\left\{#4 \middle| #5\right\}}
            }
          }
        }
      }
    }
  }
}

\newcommand{\rv}[1]{\ensuremath{\mathsf{\uppercase{#1}}}}
\newcommand{\rvt}[1]{\vt{\rv{#1}}} 
 \newcommand{\wrvt}[1]{\wt{\rvt{#1}}} 
\newcommand{\hrv}[1]{\hat{\rv{#1}}}  

\DeclareDocumentCommand{\mutInf}{ O{\empty} m m}{
  \ensuremath{\mathrm{I}
    \ifthenelse{\equal{#1}{big}}{\big({#2; #3}\big)}{
      \ifthenelse{\equal{#1}{Big}}{\Big({#2; #3}\Big)}{
        \ifthenelse{\equal{#1}{bigg}}{\bigg({#2; #3}\bigg)}{
          \ifthenelse{\equal{#1}{Bigg}}{\Bigg({#2; #3}\Bigg)}{
            \ifthenelse{\equal{#1}{normal}}{({#2; #3})}{\left({#2; #3}\right)}
          }
        }
      }
    }
  }
}

\DeclareDocumentCommand{\condMutInf}{ O{\empty} m m m }{
  \ensuremath{\mathrm{I}
    \ifthenelse{\equal{#1}{big}}{\big(#2; #3 \big| #4\big)}{
      \ifthenelse{\equal{#1}{Big}}{\Big(#2; #3 \Big| #4\Big)}{
        \ifthenelse{\equal{#1}{bigg}}{\bigg(#2; #3 \bigg| #4\bigg)}{
          \ifthenelse{\equal{#1}{Bigg}}{\Bigg(#2; #3 \Bigg| #4\Bigg)}{
            \ifthenelse{\equal{#1}{normal}}{(#2; #3 | #4)}{\left(#2; #3 \middle| #4\right)}
          }
        }
      }
    }
  }
}

\DeclareDocumentCommand{\wtyp}{ O{\empty} O{\empty} m}{
  \ensuremath{\mathcal{A}^{\ifthenelse{\equal{#1}{}}{}{(#1)}}_{#2}(#3)}
}
\DeclareDocumentCommand{\condWtyp}{ O{\empty} O{\empty} m m m}{
  \ensuremath{\mathcal{A}^{\ifthenelse{\equal{#1}{}}{}{(#1)}}_{#2}{(#3 | #4 = #5)}}
}

\DeclareDocumentCommand{\typ}{ O{\empty} O{\empty} m}{
  \ensuremath{\mathcal{T}^{#1}_{[#3]#2}}
}
\DeclareDocumentCommand{\type}{ O{\empty} m}{
  \ensuremath{\mathcal{T}^{#1}_{#2}}
}
\DeclareDocumentCommand{\condTyp}{ O{\empty} O{\empty} m m m}{
  \ensuremath{\mathcal{T}^{#1}_{[#3|#4]#2}(#5)}
}
\DeclareDocumentCommand{\condType}{ O{\empty} m m m}{
  \ensuremath{\mathcal{T}^{#1}_{#2|#3}(#4)}
}

\DeclareDocumentCommand{\ent}{ O{\empty} m}{
  \ensuremath{\mathrm{H}
    \ifthenelse{\equal{#1}{big}}{\big(#2\big)}{
      \ifthenelse{\equal{#1}{Big}}{\Big(#2\Big)}{
        \ifthenelse{\equal{#1}{bigg}}{\bigg(#2\bigg)}{
          \ifthenelse{\equal{#1}{Bigg}}{\Bigg(#2\Bigg)}{
            \ifthenelse{\equal{#1}{normal}}{(#2)}{\left(#2\right)}
          }
        }
      }
    }
  }
}

\DeclareDocumentCommand{\entRate}{ O{\empty} m}{
  \ensuremath{\ol{\mathrm{H}}
    \ifthenelse{\equal{#1}{big}}{\big(#2\big)}{
      \ifthenelse{\equal{#1}{Big}}{\Big(#2\Big)}{
        \ifthenelse{\equal{#1}{bigg}}{\bigg(#2\bigg)}{
          \ifthenelse{\equal{#1}{Bigg}}{\Bigg(#2\Bigg)}{
            \ifthenelse{\equal{#1}{normal}}{(#2)}{\left(#2\right)}
          }
        }
      }
    }
  }
}

\DeclareDocumentCommand{\entPhi}{ O{\empty} m}{
  \ensuremath{\mathrm{H}_{\phi}{}%
    \ifthenelse{\equal{#1}{big}}{\big(#2\big)}{
      \ifthenelse{\equal{#1}{Big}}{\Big(#2\Big)}{
        \ifthenelse{\equal{#1}{bigg}}{\bigg(#2\bigg)}{
          \ifthenelse{\equal{#1}{Bigg}}{\Bigg(#2\Bigg)}{
            \ifthenelse{\equal{#1}{normal}}{(#2)}{\left(#2\right)}
          }
        }
      }
    }
  }
}

\DeclareDocumentCommand{\condEnt}{ O{\empty} m m}{
  \ensuremath{\mathrm{H}
    \ifthenelse{\equal{#1}{big}}{\big({#2\big| #3}\big)}{
      \ifthenelse{\equal{#1}{Big}}{\Big({#2\Big| #3}\Big)}{
        \ifthenelse{\equal{#1}{bigg}}{\bigg({#2\bigg| #3}\bigg)}{
          \ifthenelse{\equal{#1}{Bigg}}{\Bigg({#2\Bigg| #3}\Bigg)}{
            \ifthenelse{\equal{#1}{normal}}{({#2| #3})}{\left(#2 \middle| #3\right)}
          }
        }
      }
    }
  }
}

\DeclareDocumentCommand{\dent}{ O{\empty} m}{
  \ensuremath{\mathrm{h}
    \ifthenelse{\equal{#1}{big}}{\big(#2\big)}{
      \ifthenelse{\equal{#1}{Big}}{\Big(#2\Big)}{
        \ifthenelse{\equal{#1}{bigg}}{\bigg(#2\bigg)}{
          \ifthenelse{\equal{#1}{Bigg}}{\Bigg(#2\Bigg)}{
            \ifthenelse{\equal{#1}{normal}}{(#2)}{\left(#2\right)}
          }
        }
      }
    }
  }
}

\DeclareDocumentCommand{\dentRate}{ O{\empty} m}{
  \ensuremath{\ol{\mathrm{h}}
    \ifthenelse{\equal{#1}{big}}{\big(#2\big)}{
      \ifthenelse{\equal{#1}{Big}}{\Big(#2\Big)}{
        \ifthenelse{\equal{#1}{bigg}}{\bigg(#2\bigg)}{
          \ifthenelse{\equal{#1}{Bigg}}{\Bigg(#2\Bigg)}{
            \ifthenelse{\equal{#1}{normal}}{(#2)}{\left(#2\right)}
          }
        }
      }
    }
  }
}

\DeclareDocumentCommand{\condDent}{ O{\empty} m m}{
  \ensuremath{\mathrm{h}
    \ifthenelse{\equal{#1}{big}}{\big({#2\big| #3}\big)}{
      \ifthenelse{\equal{#1}{Big}}{\Big({#2\Big| #3}\Big)}{
        \ifthenelse{\equal{#1}{bigg}}{\bigg({#2\bigg| #3}\bigg)}{
          \ifthenelse{\equal{#1}{Bigg}}{\Bigg({#2\Bigg| #3}\Bigg)}{
            \ifthenelse{\equal{#1}{normal}}{({#2| #3})}{\left(#2 \middle| #3\right)}
          }
        }
      }
    }
  }
}

\newcommand{\binEntOp}{\mathrm{H}_2}

\DeclareDocumentCommand{\binEnt}{ O{\empty} m}{
  \ensuremath{\binEntOp{
      \ifthenelse{\equal{#2}{\empty}}{}{
        \ifthenelse{\equal{#1}{big}}{\big({#2}\big)}{
          \ifthenelse{\equal{#1}{Big}}{\Big({#2}\Big)}{
            \ifthenelse{\equal{#1}{bigg}}{\bigg({#2}\bigg)}{
              \ifthenelse{\equal{#1}{Bigg}}{\Bigg({#2}\Bigg)}{
                \ifthenelse{\equal{#1}{normal}}{({#2})}{\left(#2\right)}
              }
            }
          }
        }
      }
    }
  }
}

\DeclareDocumentCommand{\binEntInv}{ O{\empty} m}{
  \ensuremath{\binEntOp^{-1}{
      \ifthenelse{\equal{#2}{\empty}}{}{
        \ifthenelse{\equal{#1}{big}}{\big({#2}\big)}{
          \ifthenelse{\equal{#1}{Big}}{\Big({#2}\Big)}{
            \ifthenelse{\equal{#1}{bigg}}{\bigg({#2}\bigg)}{
              \ifthenelse{\equal{#1}{Bigg}}{\Bigg({#2}\Bigg)}{
                \ifthenelse{\equal{#1}{normal}}{({#2})}{\left(#2\right)}
              }
            }
          }
        }
      }
    }
  }
}

\def\barcirci{{%
    \setbox0\hbox{\ensuremath{\relbar\!\!\relbar}}%
    \rlap{\hbox to \wd0{\hss\ensuremath{\circ}\hss}}\box0
}}

\newcommand{\gauss}[2]{\ensuremath{{\mathcal{N}(#1,#2)}}}
\newcommand{\bernoulli}[1]{\ensuremath{{\BBB(#1)}}}
\DeclareDocumentCommand{\uniform}{ O{\empty} m}{
  \ensuremath{\mathcal{U}
    \ifthenelse{\equal{#1}{big}}{\big({#2}\big)}{
      \ifthenelse{\equal{#1}{Big}}{\Big({#2}\Big)}{
        \ifthenelse{\equal{#1}{bigg}}{\bigg({#2}\bigg)}{
          \ifthenelse{\equal{#1}{Bigg}}{\Bigg({#2}\Bigg)}{
            \ifthenelse{\equal{#1}{normal}}{({#2})}{\left(#2\right)}
          }
        }
      }
    }
  }
}

\newcommand{\orvt}[1]{\ensuremath{\ol{\rvt{#1}}}}

\newcommand{\set}[1]{\ensuremath{\mathcal{#1}}}

\newcommand{\BBB}{\ensuremath{\set{B}}}

\newcommand{\DDD}{\ensuremath{\set{D}}}
\newcommand{\EEE}{\ensuremath{\set{E}}}

\newcommand{\PPP}{\ensuremath{\set{P}}}

\newcommand{\XXX}{\ensuremath{\set{X}}}
\newcommand{\YYY}{\ensuremath{\set{Y}}}

\newcommand{\eps}{\ensuremath{\varepsilon}}

\newcommand{\thetafkt}[1][\empty]{\ensuremath{\hat{\theta}(\ifthenelse{\equal{#1}{\empty}}{\rho,\aii,\bii}{#1})}}

\newcommand{\aii}{\ensuremath{\alpha}}

\newcommand{\bii}{\ensuremath{\beta}}

\newcommand{\CST}[1][\empty]{\ensuremath{C_{\ifthenelse{\equal{#1}{\empty}}{\aii,\bii}{#1}}}}

\input{tex/custom_versions}

\ifdefined\tikzextpath
\usetikzlibrary{external}
\fi

\pgfplotsset{compat=1.14}

\usetikzlibrary{arrows}
\usetikzlibrary{arrows.meta}
\usetikzlibrary{decorations}
\usetikzlibrary{backgrounds}
\usetikzlibrary{fit}
\usetikzlibrary{calc}
\usetikzlibrary{intersections}
\usetikzlibrary{decorations}
\usetikzlibrary{spy}
\usetikzlibrary{positioning}
\usetikzlibrary{shapes}
\usetikzlibrary{graphs}
\usetikzlibrary{shapes.multipart}
\usetikzlibrary{automata}
\usetikzlibrary{fixedpointarithmetic}
\usetikzlibrary{decorations.markings}

\definecolor{myred}{RGB}{128, 0, 0}
\definecolor{myblue}{RGB}{0, 0, 128}

\tikzstyle{block} = [draw,fill=RoyalBlue!30,minimum size=2em,rounded corners=2mm]
\tikzstyle{symb}=[]
\def\radius{.8mm} 
\tikzstyle{branch}=[fill,shape=circle,minimum size=3pt,inner sep=0pt]
\tikzstyle{s}=[shift={(0mm,\radius)}]

\makeatletter
\newcommand{\includetikz}[1]{%
  \ifdefined\tikzexternalize%
  \filename@parse{#1}%
  \tikzsetnextfilename{\filename@base}%
  \fi%
  \input{#1.tikz}%
}
\makeatother

\usepackage[acronym,nohypertypes={acronym}]{glossaries}
\setacronymstyle{long-short}

\usepackage{authblk}

\newcommand{\wvt}[1]{\ensuremath{\wt{\vt{#1}}}}

\newcommand{\lebesgue}{\lambda}
\newcommand{\Ball}{B}
\usetikzlibrary{arrows,chains,matrix,positioning,scopes}

\newacronym{pdf}{pdf}{probability density function}
\newacronym{dgc}{DGC}{dense graph condition}
\newacronym{rhs}{RHS}{right hand side}

\glsunset{rhs}

\allowdisplaybreaks[4]
\usepackage{autonum}

\def\mytitle{On the Estimation of Information Measures of Continuous Distributions}

\begin{document}

\let\citealt\cite


\title{\mytitle}
\author[1]{Georg Pichler}
\author[2]{Pablo Piantanida}
\author[3]{Günther Koliander}
\affil[1]{Institute of Telecommunications, TU Wien, Vienna, Austria}
\affil[2]{Universit\'e Paris-Saclay, CNRS, CentraleSup\'elec, Laboratoire des Signaux et Syst\`emes,  Gif-sur-Yvette, France and Montreal Institute for Learning Algorithms (Mila), Universit\'e de Montr\'eal, QC, Canada}
\affil[3]{Acoustics Research Institute, Austrian Academy of Sciences, Vienna, Austria}
\date{}                     
\setcounter{Maxaffil}{0}
\renewcommand\Affilfont{\itshape\small}

\maketitle

\section*{Abstract}
The estimation of information measures of continuous distributions based on samples is a fundamental problem in statistics and machine learning. In this paper, we analyze estimates of differential entropy in $K$-dimensional Euclidean space, computed from a finite number of samples, when the probability density function belongs to a predetermined convex family $\PPP$.
First, estimating differential entropy to any accuracy is shown to be infeasible if the differential entropy of densities in $\PPP$ is unbounded, clearly showing the necessity of additional assumptions.
Subsequently, we investigate sufficient conditions that enable confidence bounds for the estimation of differential entropy.
In particular, we provide confidence bounds for simple histogram based estimation of differential entropy from a fixed number of samples, assuming that the probability density function is Lipschitz continuous with known Lipschitz constant and known, bounded support. Our focus is on differential entropy, but we provide examples that show that similar results hold for mutual information and relative entropy as well.

\def\acknowledgments{
  \section*{Acknowledgments}
  The authors are very grateful to Prof.\ Elisabeth Gassiat for pointing out the connection between the present work and the reference~\cite{Donoho1988One}.
  This project has received funding from the European Union’s Horizon 2020 research and innovation programme under the Marie Skłodowska-Curie grant agreement No 792464.
}

\section{Introduction}
\label{sec:introduction}

Many learning tasks, especially in unsupervised/semi-supervised settings, use information theoretic quantities, such as relative entropy, mutual information, differential entropy, or other divergence functionals as target functions in numerical optimization problems~\cite{Hjelm2018Learning,Hu2017Learning,Miyato2015Distributional,Oord2018Representation,Gordon2003Applying,Chen2016InfoGAN,Barber2003IM}. Furthermore, estimators for information theoretical quantities are useful in other fields, such as neuroscience \cite{Nemenman2004Entropy}.
As these quantities typically cannot be computed directly, surrogate functions, either upper/lower bounds, or estimates are used in place. Here, we will investigate the problem of estimating differential entropy using a finite number of samples. Throughout, we will restrict our attention to differential entropy, but similar results also hold for conditional differential entropy, mutual information and relative entropy (\cf \cref{sec:estim-other-meas}).

\subsection{Our Contribution}

The contributions of this work can be summarized as follows: 
\begin{itemize}
\item First, we explore the following basic but fundamental question: Fixing $C, \delta > 0$ and given $N \in \NN$ samples from a \gls{pdf} $\p{} \in \PPP$, where $\PPP$ is a family of \glspl{pdf} on $\RR^K$, is it possible to obtain an estimate $\hat h$ of the differential entropy $\dent{\p{}} < \infty$ satisfying 
\begin{equation}
\Prob{ |\hat h - \dent{\p{}}| > C} \le \delta\,?
\end{equation}
In~\cref{sec:estim-diff-entr}, we show that the answer to this question is negative (\cref{pro:diff_entropy_estimation}) if $\PPP$ is convex and the differential entropy of the \glspl{pdf} in $\PPP$ is unbounded.
\item  Subsequently, we investigate sufficient conditions for the class $\PPP$ that enable estimation of differential entropy with such a confidence bound and in~\cref{sec:lipschitz} (\cref{thm:lipschitz}) we show that a known, bounded support together with an $L$-Lipschitz continuous \gls{pdf} for fixed $L > 0$, suffices.\footnote{These assumptions assure that the differential entropy of the \glspl{pdf} in $\PPP$ is bounded. A known, bounded support bounds the differential entropy from above, and $L$-Lipschitz continuity bounds it from below.} For a simple histogram based estimator we explicitly compute a relation between probability of correct estimation, accuracy, dimension $K$, sample size $N$, and Lipschitz constant $L$. It is shown that estimation becomes impossible if either assumption is removed.

\item Finally, in~\cref{sec:estim-other-meas} we obtain impossibility results, similar to~\cref{pro:diff_entropy_estimation}, for the estimation of other information measures. 

\end{itemize}

\subsection{Previous Work}
\label{sec:previous-work}

The problem of estimating information measures from a finite number of samples is as old as information theory itself. Shortly after his seminal paper \cite{Shannon1948Mathematical}, Shannon worked on estimating the entropy rate of English text \cite{Shannon1951Prediction}. There have been numerous works on the estimation of information measures, such as entropy, mutual information, and differential entropy, since. There are many different approaches for estimating information measures, including kernel based methods, nearest neighbor methods, methods based on sample distances as well as multiple variants of plug-in estimates. Many estimators have been shown to be consistent and/or asymptotically unbiased under various constraints, \eg, in~\cite{Gyoerfi1987Density,Ahmad1976Nonparametric,Gao2017Estimating,Kandasamy2015Nonparametric,Sricharan2011k}. An excellent overview can be found in~\cite{Beirlant1997Nonparametric}.

In~\cite{Sricharan2011k}, rate-of-convergence results as well as a central limit theorem are provided for differential entropy and Rényi entropy. However, the confidence bounds and the constants involved in the rate-of-convergence results depend on the underlying distribution which is typically unknown.
Similarly, \cite{Han2017Optimal} obtains a rate-of-convergence result, assuming a Lipschitz ball smoothness assumption combined with known compact support, but the involved constants remain unspecified. In a similar spirit, \cite{Liu2012Exponential} provides asymptotic results for the estimation of differential entropy in two dimensions, when certain smoothness conditions are satisfied and the \gls{pdf} is bounded away from zero.
The related task of estimating relative entropy is studied, \eg, in \cite{Wang2005Divergence,Nguyen2010Estimating} and partition-based estimation of mutual information is analyzed in \cite{Darbellay1999Estimation}. While \cite{Wang2005Divergence} only shows consistency, convergence rates are obtained in \cite[Th.~2]{Nguyen2010Estimating}, but again the constants involved remain unspecified.

In contrast to our present work, the existing results for the estimation of differential entropy mentioned above fall short when addressing the practical problem of a finite sample size.
However, some results are available in a more general context. In \cite{Singh2016Finite}, a finite-sample analysis is conducted. Similar to our approach (\cf \cref{sec:lipschitz}), the authors of \cite{Singh2016Finite} assume a fixed support $[0,1]^K$, but instead of Lipschitz continuity, $\beta$-Hölder continuity, $\beta \in (0, 2]$ is assumed. Additionally, strict positivity on the interior of the support is required and the constants bounding the approximation error depend on the underlying, unknown distribution. 
These additional complications are likely due to the extended scope, as \cite{Singh2016Finite} is not focused on differential entropy, but the expectation of arbitrary functionals of the probability density.
The same authors also provide finite sample analysis for the estimation of Rényi divergence under similarly strong conditions in~\cite{Singh2014Generalized}.

There are several negative results, which clearly show that information measures are hard to estimate from a finite number of samples.
It was shown in \cite[Th.~4]{Antos2001Convergence} that rate-of-convergence results cannot be obtained for any consistent estimator of entropy on a countable alphabet and only when imposing various assumptions on the true distribution, rate-of-convergence results were obtained. More negative results on the estimation of entropy and mutual information can be found in \cite{Paninski2003Estimation}.
In fact, obtaining confidence bounds for information measures from samples is inherently difficult and requires regularity assumptions about the involved distributions, which are not subject to empirical test. In the seminal work of \cite{Bahadur1956Nonexistence} as well as subsequent works \cite{Gleser1987Nonexistence,Donoho1988One,Pfanzagl1998Nonexistence,Donoho1991Geometrizing} (and references therein) such necessary conditions for the estimation of statistical parameters with confidence bounds are discussed in great detail and generality.
The results of \cite{Donoho1988One,Pfanzagl1998Nonexistence} can be applied to differential entropy estimation and yield a result very similar to \cref{pro:diff_entropy_estimation}, essentially showing that differential entropy cannot be bounded using a finite number of samples, unless additional assumptions on the distribution are made.

Especially in the context of unsupervised and semi-su\-per\-vised machine learning it recently became popular to use variational bounds or estimates of information measures as part of the loss function for training neural networks \cite{Belghazi2018Mutual,Poole2019Variational,Hjelm2018Learning}.
Criticism to this approach, in particular the use of variational bounds, has been voiced \cite{McAllester2018Formal}. The current paper has a more general scope, dealing with the estimation problem of information measures in general, not limited to specific variational  bounds or techniques.

The information flow in neural networks is also a recent topic of investigation. In~\cite{Shwartz-Ziv2017Opening}, an argument for successive compression in the layers of a deep neural network is given, along the lines of the information bottleneck method~\cite{Tishby2000Information}. While flaws in this argument were pointed out \cite{Amjad2019Learning,Kolchinsky2018Caveats}, the authors of \cite{Goldfeld2019Estimating} found that a clustering phenomenon might elucidate the behavior of deep neural networks. These insights were obtained by estimating the differential entropy $\dent{\rvt x + \rvt g}$ of a sum of two random vectors, where $\rvt x$ is sub-Gaussian and $\rvt g \sim \gauss{\vt 0}{\sigma^2 \vt I}$ is an independent Gaussian vector. This is similar in spirit to the work conducted here, however, our assumption of compact support is replaced by assuming $\rvt x$ to be sub-Gaussian.%
\footnote{Similar to our assumption of an arbitrary but fixed compact support in \cref{sec:lipschitz}, the constant $K$ in the definition of the sub-Gaussian $\rvt x$ is assumed to be fixed in \cite[eq.~(1)]{Goldfeld2019Convergence}.}
Note that the \gls{pdf} of $\rvt x + \rvt g$ is Lipschitz continuous with fixed Lipschitz constant $L(\sigma^2)$, so \cite{Goldfeld2019Estimating} is implicitly also using a Lipschitz assumption.

\section{The Nonexistence of Confidence Sets}
\label{sec:estim-diff-entr}

Let $\PPP$ be a family of \glspl{pdf} on $\XXX \defas \RR^K$ with finite differential entropy, \ie, $\dent{\p{}} \defas -\int \p{\vt x} \log\p{\vt x} \,\dd \vt x \in \RR$ for every $\p{} \in \PPP$.

Suppose we observe $N$ \iid copies $\DDD \defas (\rvt x_1, \rvt x_2, \dots, \rvt x_N)$ of some random vector $\rvt x \sim \p{} \in \PPP$ and want to obtain an estimate of differential entropy from these samples $\DDD$.
Such an estimator is a function $\hat h \colon \XXX^N \to \RR$ that maps $\DDD$ into $\hat h(\DDD)$, approximating the differential entropy $h \defas \dent{\rvt x} \defas \dent{\p{}} < \infty$.
Its accuracy can be measured by a confidence interval, a widely used tool in statistical practice for indicating the precision of point estimators.
For a given error probability $\delta > 0$, we would like to have $C > 0$ such that $|h - \hat h(\DDD)| \ge C$ with probability less than $\delta$, \ie, a confidence interval of size $C$ with confidence $1 - \delta$.
However, there is no free lunch when estimating differential entropy, as evidenced by the following result, a corollary of a more general result in \cite{Donoho1988One}, here specialized to a bound of differential entropy. It is based on the abstract notion of a \gls{dgc}.\footnote{$\mathrm h$ satisfies the \gls{dgc} over $\PPP$ if the graph of $\mathrm h$ over $\PPP$ is dense in its own epigraph \cite[eq.~(2.4)]{Donoho1988One}.}
\begin{theorem}[{\citealt[Th.~2.1]{Donoho1988One}}]
  \label{thm:donoho}
    Assume that $\mathrm{h}$ satisfies the \gls{dgc} over $\PPP$ and define $B \defas \sup\{\dent{\p{}} : \p{} \in \PPP\}$, where, \eg, $B = +\infty$.
  If 
  for any $C>0$, 
  $\sup_{\p{} \in \PPP} \Prob[\p{}][normal]{\hat h(\DDD) 
  +C
  \le B} = 1$, then
  \begin{align}
    \inf_{\p{} \in \PPP} \Prob[\p{}]{\dent{\p{}} - \hat h(\DDD) \le C} = 0 .
  \end{align}
\end{theorem}
A similar result follows from ~\cite[Prop.~3.1]{Pfanzagl1998Nonexistence}.

We will not work with the \gls{dgc}, but make two practical assumptions: $\PPP$ is a convex family and the differential entropy of the \glspl{pdf} in $\PPP$ is unbounded (either from above or from below).
Under these assumptions we show that for any $\delta, C > 0$ there is a \gls{pdf} $\p{} \in \PPP$, such that $|h - \hat h(\DDD)| \le C$ with probability less than $\delta$, \ie, $\hat h(\DDD)$ is far from $h$ with high probability. Fundamentally, this follows from the fact that $\PPP$ contains \glspl{pdf} with a large difference in differential entropy, which cannot be accurately distinguished based on samples.
Similar results hold true for mutual information and relative entropy and are given in \cref{sec:estim-other-meas}.

\begin{proposition}
  \label{pro:diff_entropy_estimation}
  Let $\PPP$ be a convex family of \glspl{pdf} with unbounded differential entropy, \ie, for any $\alpha \in [0,1]$ and $\p{},\q{} \in \PPP$, 
  we have $\alpha \p{} + (1-\alpha)\q{} \in \PPP$ as well as $\sup_{\q{} \in \PPP} |\dent{\q{}}| = \infty$.
    Then, for any pair of constants $C, \delta > 0$, there exists a continuous random vector $\rvt x \sim \p{} \in \PPP$, satisfying
  \begin{equation}
    \Prob[\p{}]{| \dent{\p{}} - \hat h(\DDD) | \le C} \le \delta .
  \end{equation}
\end{proposition}
\begin{remark}
  \label{rmk:donoho}
  Before proceeding with the proof of \cref{pro:diff_entropy_estimation}, we note that this result could be proved as a consequence of \cref{thm:donoho}. 
  However, this would necessitate  to show that our conditions imply the \gls{dgc}.
    Furthermore, the proof of \cite[Th.~2.1]{Donoho1988One} itself hinges on deep statistical results and thus we opted for providing a short, self-contained proof.
\end{remark}
\begin{proof}[Proof of \cref{pro:diff_entropy_estimation}]
  The function $\hat h$, constants $C, \delta > 0$, and the sample size $N \in \NN$ are arbitrary, but fixed.
  Choose an arbitrary $\orvt x \sim \op{} \in \PPP$ and let $\ol h \defas \dent{\op{}} < \infty$. Then fix $b > 0$, such that $\Prob{|\hat h(\orvt x_1, \orvt x_2, \dots, \orvt x_N) - \ol h| \le b} \ge 1-\frac{\delta}{2}$, where $(\orvt x_n)_{n=1,\dots,N}$ are \iid copies of $\orvt x$.
  Furthermore, let $\rv q \sim \bernoulli{1-\eps}$ be a Bernoulli random variable with parameter $1-\eps = \Prob{\rv q = 1}$, independent of $\orvt x$, where $0 < \eps \le \frac{\delta}{2N}$. Choose $a > 0$ such that $a\eps > b + C + \log 2$.
    By our assumption $\sup_{\q{} \in \PPP} |\dent{\q{}}| = \infty$, we can
  find $\wrvt x \sim \wtp{} \in \PPP$ with $\wt h \defas \dent{\wtp{}}$ such that $|\wt h - \ol h| \ge a$.

  Define
  $
  \rvt x \defas \rv q \orvt x - (1-\rv q)\wrvt x
  $, which yields $h = \dent{\rvt x} = \mutInf{\rvt x}{\rv q} + \condDent{\rvt x}{\rv q} = \mutInf{\rvt x}{\rv q} + (1-\eps) \ol h + \eps \wt h \in (1-\eps) \ol h + \eps \wt h + [0,\log 2]$ where $\mutInf{\wc}{\wc}$ denotes mutual information.
   For convenience we use $\hrv h \defas \hat h(\DDD)$, where $\DDD = (\rvt x_1, \rvt x_2, \dots, \rvt x_N)$ and define the event $\EEE = \{\rv q_1 = \rv q_2 = \cdots = \rv q_N = 1\} = \{\DDD = (\orvt x_1, \orvt x_2, \dots, \orvt x_N)\}$. By the union bound, we have $\Prob{\EEE} \ge 1 - N\eps$ and obtain
   \begin{align}
     & \Prob{|\hrv h - \ol h| \le b} \eqnl
     &\qquad = \Prob{\EEE} \Pcond{|\hrv h - \ol h| \le b}{\EEE} \eqnl
     &\qquad\quad + \Prob{\EEE\compl} \Pcond{|\hrv h - \ol h| \le b}{\EEE\compl} \\
     &\qquad \ge (1-N\eps) \Pcond{|\hrv h - \ol h| \le b}{\EEE}  \\
     &\qquad = (1-N\eps) \Prob{|\hat h(\orvt x_1, \orvt x_2, \dots, \orvt x_N) - \ol h| \le b}  \\
     &\qquad \ge \left(1-\frac{\delta}{2}\right)^2 
       \ge 1-\delta . 
  \end{align}
  We thus found $\rvt x$ such that
  \begin{align}
    \label{eq:1}
    & \Prob{|h - \hrv h| \le C}
    \\*
    & \qquad \le \Prob{|\ol h - \hrv h|  \ge |\ol h - h| - C} \\
    & \qquad \le \Prob{|\ol h - \hrv h|  \ge \eps |\ol h - \wt h| - \log 2 - C} \\
    & \qquad \le \Prob{|\ol h - \hrv h|  \ge \eps a - \log 2 - C} \\
    & \qquad \le \Prob{|\ol h - \hrv h|  > b} \\
    & \qquad = 1 - \Prob{|\ol h - \hrv h|  \le b} \le \delta . 
    \qedhere
  \end{align}
\end{proof}

\begin{remark}
  \Cref{pro:diff_entropy_estimation} shows that in order to obtain confidence bounds, one needs to make assumptions about the underlying distribution. However, as pointed out in \cite[p.~1395]{Donoho1988One}, when making these assumptions, one uses information external to the samples.
\end{remark}

\begin{remark}
  \label{rmk:donoho_example}
  Note that the family of all \glspl{pdf} with support $[0,1]^K$ satisfies the requirements of \cref{pro:diff_entropy_estimation}. It also satisfies the \gls{dgc}, but it is not strongly nonparametric, as defined in \cite[p.~1395]{Donoho1988One}.
\end{remark}

\section{Lipschitz Density Assumption }
\label{sec:lipschitz}

One way to avoid the problems outlined in \cref{sec:estim-diff-entr} is to impose additional assumptions on the underlying probability distribution, that bound the differential entropy from above and from below.
We will showcase that the differential entropy of an $L$-Lipschitz continuous \gls{pdf} with fixed, known $L > 0$ on $\RR^K$ and known, compact support $\XXX$ can be well approximated from samples.
In the following, let $\rvt x \sim \p{}$ be supported\footnote{Any known compact support suffices. An affine transformation then yields $\XXX = [0,1]^K$, while possibly resulting in a different Lipschitz constant.} on $\XXX \defas [0,1]^K$, \ie, $\int_{\XXX} \p{} \,\dd\lebesgue^K = \Prob{\rvt x \in \XXX} = 1$, where $\lebesgue^K$ denotes the Lebesgue measure on $\RR^K$.
The \gls{pdf} $\p{}\colon \RR^K \to \RR_+$ of $\rvt x$ is assumed to be $L$-Lipschitz continuous on $\RR^K$ with some fixed $L > 0$, where $\RR^K$ is equipped with the $\ell_1$-norm\footnote{The $L_1$ norm is only chosen to facilitate subsequent computations. By the equivalence of norms on $\RR^K$, any norm suffices.} $\norm{\vt x} \defas \norm{\vt x}_1 = \sum_k{|x_k|}$, hence,
\begin{align}
  \forall \vt x, \vt y \in \RR^K : |\p{\vt x} - \p{\vt y}| \le L \norm{\vt x - \vt y} .
\end{align}

Given $N$ \iid copies $\DDD = (\rvt x_1, \rvt x_2, \dots, \rvt x_N)$ of $\rvt x$, let $\rvt y$ be distributed according to the empirical distribution of $\DDD$, \ie, $\rvt y = \rvt x_{\rv u}$, where $\rv u \sim \uniform{\Nto{N}}$ is a uniform random variable on $\Nto{N}$. Let the discrete random vector $\wrvt y = \Delta_M(\rvt y)$ be the element-wise quantization of $\rvt y$, where $\Delta_M(x) \defas \frac{\lfloor Mx \rfloor}{M}$ is the $M$-step discretization of $[0,1]$ for some $M \in \NN$.
Additionally define the continuous random vector $\orvt y = \wrvt y + \uniform{[0,\frac 1M]^K}$, \ie, independent uniform noise is added. Note also that $\dent[normal]{\orvt y} - \ent[normal]{\wrvt y} = - K \log M$, where $\ent{\wc}$ denotes Shannon entropy.

We will estimate differential entropy by $\hat h(\DDD) = \dent[norma]{\orvt y} = \ent[normal]{\wrvt y} - K \log M$, \ie, the Shannon entropy of the discretized and binned samples with a correction factor.

In the following, we shall also use the two constants
\pgfmathsetmacro{\myalpha}{( (sqrt(exp(2)+4)-exp(1)) )/(2*exp(1))}
\begin{align}
  \eta(K,L) &\defas  \frac{1}{K} \left(\frac{2 (K+1)!}{L}\right)^{\frac{1}{K+1}} \text{ , and} \\
  \alpha &\defas \frac{\sqrt{\ee{2} + 4} - \ee{}}{2 \ee{}} \approx \pgfmathprintnumber[fixed,precision=5]{\myalpha} .
\end{align}

\begin{theorem}
  \label{thm:lipschitz}
  For $M \ge \frac{1}{\alpha \eta(K,L)}$ and any $\delta \in (0,1)$, we have with probability greater than $1 - \delta$ that
  \begin{subequations}
    \label{eq:lipschitz}
    \begin{align}
      \big|\hat h(\DDD) - \dent{\rvt x} \big| 
      & \le \frac{LK}{2M} \log(M \eta(K,L))
        \label{eq:lipschitz:bias1}
      \\
      & \quad
        + \sqrt{\frac{2}{N} \log\frac{2}{\delta}} \log N \label{eq:lipschitz:var}
      \\
    & \quad
      + \log\left( 1 + \frac{M^K-1}{N} \right) .
      \label{eq:lipschitz:bias2}
    \end{align}
  \end{subequations}
\end{theorem}
The proof will be given in \cref{sec:proof:lipschitz}.

\begin{remark}
  Of the three error terms \cref{eq:lipschitz:bias1,eq:lipschitz:var,eq:lipschitz:bias2}, the terms \cref{eq:lipschitz:bias1,eq:lipschitz:bias2} constitute the bias and \cref{eq:lipschitz:var} is a variance-like error term. While the variance \cref{eq:lipschitz:var} vanishes as $N \to \infty$, the term \cref{eq:lipschitz:bias1} does not depend on the sample size $N$ as it merely measures the error incurred due to the quantization $\Delta_M$, which is bounded by the Lipschitz constraint and approaches zero as $M \to \infty$. The final term \cref{eq:lipschitz:bias2} results from ensuring that $N$ samples suffice to suitably approximate the empirical distribution over $M$ quantization steps. Thus, it ties the quantization to the sample size and approaches zero if $\frac{M^K}{N} \to 0$.
  In total, the \gls{rhs} of \cref{eq:lipschitz} approaches zero for $N,M \to \infty$ provided that $\frac{M^K}{N} \to 0$.
\end{remark}

\begin{remark}
  \label{rmk:lipschitz_not_tight}
  \Cref{thm:lipschitz} should be regarded as a proof-of-concept rather than a practical tool for performing differential entropy estimation. While analytically tractable, the estimation strategy is crude and the bounds, especially the term \cref{eq:lipschitz:var}, while being a completely universal bound, is know to be loose, as pointed out in \cite[p.~1200]{Paninski2003Estimation}.
\end{remark}

\begin{remark}
  \label{rmk:assumptions_necessity}
  We want to note that requiring both a fixed Lipschitz constant $L$ and a known bounded support, \eg, $\XXX = [0,1]^K$, is necessary. Consider for instance the set $\PPP' = \{\p{} : \p{}$ supported on $[0,1]^K$  and Lipschitz continuous$\}$ of \glspl{pdf} with arbitrary Lipschitz constant or the set $\PPP'' = \{\p{} : \p{}$ supported on a bounded set and $L$-Lipschitz continuous$\}$ with fixed Lipschitz constant, but arbitrary, bounded support. Both families satisfy the conditions of \cref{pro:diff_entropy_estimation}, \ie, they are convex and $\sup_{\p{} \in \PPP'} |\dent{\p{}}| = \sup_{\p{} \in \PPP''} |\dent{\p{}}| = \infty$.
\end{remark}

In principle, \cref{thm:lipschitz} also allows for the approximation of mutual information with a confidence bound. Let $(\rvt x, \rvt y) \sim \p[\rvt x \rvt y]{}$ be two random vectors, supported on $[0,1]^{K_1}$ and $[0,1]^{K_2}$, respectively. Assuming that $\p[\rvt x \rvt y]{}$ is $L$-Lipschitz continuous on $\RR^{K_1+K_2}$, it is clear that the marginals $\p[\rvt x]{}$ and $\p[\rvt y]{}$ are $L$-Lipschitz continuous as well. Thus, \cref{thm:lipschitz} can be used to approximate all three terms in
\begin{align}
  \label{eq:mut_inf}
  \mutInf{\rvt x}{\rvt y} = \dent{\rvt x} + \dent{\rvt y} - \dent{(\rvt x, \rvt y)} .
\end{align}

\section{Estimation of other Measures}
\label{sec:estim-other-meas}

In this \lcnamecref{sec:estim-other-meas}, we showcase that similar statements as \cref{pro:diff_entropy_estimation} also hold for mutual information and relative entropy.
For simplicity we will not assume $\p[\rv x \rv y]{} \in \PPP$ for some family $\PPP$ of probability density functions, but merely require $\mutInf{\rv x}{\rv y}, \DKL{\rv x}{\rv y} < \infty$. Only proof sketches are provided as the examples provided in this \lcnamecref{sec:estim-other-meas} are similar to the proof of \cref{pro:diff_entropy_estimation}.

Here we show that in general, it is not possible to accurately estimate mutual information $\mutInf{\rv x}{\rv y}$ and relative entropy $\DKL{\rv x}{\rv y}$ from samples $(\rvt x, \rvt y)$.

\subsection{Mutual Information}
\label{sec:mutual-information}

For any $N$, let $\hat i \colon \RR^N \times \RR^N \to \RR$ be a measurable function, which represents an estimate of the mutual information $I \defas \mutInf{\rv x}{\rv y} < \infty$ from $\rvt x, \rvt y$.
For convenience we use $\hrv i \defas \hat i(\rvt x, \rvt y)$. 
Let $\rv x, \rv z \sim \uniform{[0,1]}$, $\rv w \sim \uniform{[0,\ee{-a}]}$, and $\rv q \sim \bernoulli{1-\eps}$ be independent random variables. Define
\begin{align}
    \rv y &\defas \rv q \rv z - (1-\rv q) (\rv x + \rv w) .
\end{align}
We have 
\begin{align}
  \mutInf{\rv x}{\rv y} 
  &= \dent{\rv y} - \condDent{\rv y}{\rv x} \\
                                      &\ge \binEnt{\eps} - \binEnt{\eps} + a \eps   \\
                              &= a \eps .
\end{align}
The random vectors $(\rvt x, \rvt y)$ [$(\rvt x, \rvt z)$] are $N$ \iid realizations of $(\rv x, \rv y)$ [$(\rv x, \rv z)$].
For any $\delta > 0$, we can find $b \in \RR$ such that $\Prob{\hat i(\rvt x, \rvt z) \le b} \ge 1 - \frac{\delta}{2}$.
Letting $\EEE 
= \{\rvt y = \rvt z \} = \{\rvt q = \rvt 1\}$, we have $\Prob{\EEE} \ge 1 - N\eps$.
Thus, when choosing $\eps \le \frac{\delta}{2N}$,
\begin{align}
  \Prob{\hrv i \le b} &= \Prob{\EEE} \Pcond{\hrv i \le b}{\EEE} + \Prob{\EEE\compl} \Pcond{\hrv i \le b}{\EEE\compl} \\
                        &\ge (1 - N\eps) \Pcond{\hrv i \le b}{\EEE} \\
                        &\ge (1 - N\eps) \left(1 - \frac{\delta}{2}\right) \\
                        &\ge 1 - \delta .
\end{align}
We may choose $a \ge \frac{b + C}{\eps}$. Then, for arbitrary $C, \delta > 0$ and $n \in \NN$, we found $\rv x$, $\rv y$ and $b \in \RR$, such that $\Prob{\hrv i > b} \le \delta$, yet $I \ge b + C$.

\begin{remark}
  Note that \cite[Th.~3]{Belghazi2018Mutual} claims a confidence bound for mutual information, that together with the approximation result \cite[Lem.~1]{Belghazi2018Mutual} seemingly contradicts our result. However, the confidence bound proved in \cite[Th.~3]{Belghazi2018Mutual} requires strong conditions on the functions\footnote{Here we use the notation of \cite{Belghazi2018Mutual}.} $T_\theta$ and \cite[Lem.~1]{Belghazi2018Mutual} does not necessarily hold under these conditions. Moreover, both approximation results \cite[Lem.~1~and~Lem.~2]{Belghazi2018Mutual} do not hold uniformly for a family of distributions, but implicitly assume a fixed, underlying distribution.
  This is especially evident in \cite[Lem.~2]{Belghazi2018Mutual}, which also seemingly contradicts our result, when assuming that the optimal function $T^*$ is in the family $T_\Theta$. However, this apparent contradiction is resolved by noting that the chosen $N \in \NN$ depends on the underlying, true distribution.
\end{remark}

\subsubsection{One Discrete Random Variable}
\label{sec:one-discrete-random}

In the following we show that a similar result holds if $\rv y$ has a fixed finite alphabet, say $\YYY = \Nto{K}$. 
Again, for every $N \in \NN$, let $\hat i_N \colon \RR^N \times \YYY^N \to \RR$ be an estimator that estimates $I \defas \mutInf{\rv x}{\rv y} \le \log K$ from $\rvt x, \rvt y$.
Note that the result for continuous $\rv y$ cannot carry over unchanged as we have $I \in [0, \log K]$.
We shall assume that $\hat i_N$ is consistent in the sense that $\hrv i_N \to I$ in probability as $N \to \infty$, where we use $\hrv i_N \defas \hat i_N(\rvt x, \rvt y)$.

Let $\rv x \sim \uniform{[0,1]}$ and $\rv w \sim \uniform{\YYY}$ be independent random variables. Fix $\delta > 0$ and by consistency find $N_0$ such that $\Prob{\hat i_N(\rvt x, \rvt w) \ge \delta} \le \frac \delta 2$ for all $N \ge N_0$. 
In the following consider $N \ge N_0$ fixed. 
Fix $M \in \NN$ and $\vt v \in \YYY^M$, and define 
the quantization $\rv z \defas 1 + \lfloor M\rv x\rfloor$. The random variable $\rv y$ is simply $\rv y = v_{\rv z}$.
We use the notation $\hrv i_N^{\vt v}$ to highlight that $\hrv i_N$ depends on the particular choice of $\vt v$ and wish to show that $\Prob{\hrv i_N^{\vt v} \ge \delta} \le \delta$ for at least one $\vt v$.

Assume to the contrary, that $\Prob{\hrv i_N^{\vt v} \ge \delta} > \delta$ for all $\vt v \in \YYY^M$. Let $\EEE = \{\exists i,j \in \Nto{N} : i \neq j, \rv z_i = \rv z_j\}$ be the event that two elements of $\rvt x$ fall in the same ``bin.'' Note that $\rvt z$ is the quantization of $\rvt x$.
For $M$ large enough, we obtain
\begin{align}
  \Prob{\EEE} &= 1 - \frac{M!}{M^N (M-N)!} \\
              & \le 1- \left(\frac{M-N+1}{M}\right)^N \\
              &\le \eps .
\end{align}
Defining $\rvt v \sim \uniform{\YYY^M}$, independent of $\rvt x$, we obtain for $\eps$ small enough
\begin{align}
  \delta &< K^{-M} \sum_{\vt w \in \YYY^M} \Prob{\hrv i_N^{\vt v} \ge \delta} \\
                  &\le \eps + K^{-M} \sum_{\vt v \in \YYY^M} \Pcond{\hrv i_N^{\vt v} \ge \delta}{\EEE\compl} \\
                  &= \eps + \sum_{\vt v \in \YYY^M} \Pcond{\hrv i_N^{\rvt v} \ge \delta}{\EEE\compl, \rvt v = \vt v} \Prob{\rvt v = \vt v} \\
         &= \eps + \Pcond{\hrv i_N^{\rvt v} \ge \delta}{\EEE\compl} \\
         &= \eps + \Pcond{\hat i_N(\rvt x, \rvt v_{\rvt z}) \ge \delta}{\EEE\compl} \\
                           &= \eps + \Pcond{\hat i_N(\rvt x, \rvt w) \ge \delta}{\EEE\compl} \\
                  &\le \eps + \frac{\delta}{2(1-\eps)}
         \le \delta ,
\end{align}
leading to a contradiction.

To summarize, for arbitrary $\delta$ and $N$ large enough, there exists $\vt w$ such that $\Prob{\hrv i_N^{\vt w} \ge \delta} \le \delta$, but clearly $\rv y$ is a deterministic function of $\rv x$ and hence $I = \log K$.
\subsection{Relative Entropy}
\label{sec:kl}

Let $\p{}$ and $\q{}$ be two continuous \glspl{pdf} (\wrt $\lebesgue$) and $\rvt x$, $\rvt y$ be $N$ \iid random variables distributed according to $\p{}$ and $\q{}$, respectively.
For any $N$, let $\hat d_N\colon \RR^N \times \RR^N \to \RR$ be an estimator that estimates $D \defas \DKL{\p{}}{\q{}} < \infty$ from $\rvt x, \rvt y$.
For convenience we use $\hrv d_N \defas \hat d_N(\rvt x, \rvt y)$. Let $\rvt z_1$, $\rvt z_2$ be two independent \iid $N$ vectors with components uniformly distributed on $[-1,0]$.
For an arbitrary $\delta > 0$ we can find $c \in \RR$ such that $\Prob{\hat d_N(\rvt z_1, \rvt z_2) \le c} \ge 1-\frac{\delta}{2}$.
Consider $C,\delta > 0$ and an arbitrary $N \in \NN$.

Define the \glspl{pdf}
\begin{align}
    \p{x} &= \ee{-a} \ind{[0,1)}{x} + (1-\ee{-a}) \ind{[-1,0)}{x}, \text{ and} \\
    \q{x} &= \ee{-a-b} \ind{[0,1)}{x} + (1-\ee{-a-b}) \ind{[-1,0)}{x}
\end{align}
for $a,b \in \RR_+$.

With $b = k\ee{a}$, where $k \in \RR_+$, we have $D = D(a,k)$, with the function
\begin{align}
  D(a,k)                &= \ee{-a} b + (1-\ee{-a}) \log\frac{1-\ee{-a}}{1-\ee{-a-b}} \\
           &= k + (1-\ee{-a}) \log\frac{1-\ee{-a}}{1-\ee{-a-k\ee{a}}} \\
                      &\ge k - \ee{-1} .
\end{align}

Let $\EEE_1 = \{\rvt x < \vt 0\}$ and $\EEE_2 = \{\rvt y < \vt 0\}$ be the events that every component of $\rvt x$ and $\rvt y$, respectively, is negative. Then,
\begin{align}
    \Prob{\EEE_1\compl \cup \EEE_2\compl} &\le 2N\ee{-a} ,\\
  \Prob{\EEE_1 \cap \EEE_2} &\ge 1 - 2N\ee{-a} .
\end{align}
Choose $a \ge \log \frac{4N}{\delta}$ and $k \ge c + C + \ee{-1}$ such that $2N\ee{-a} \le \frac{\delta}{2}$ and $D \ge c + C$.
We can now bound the probability
\begin{align}
  \Prob{\hrv d_N \le c} &= \Pcond{\hrv d_N \le c}{\EEE_1 \cap \EEE_2} \Prob{\EEE_1 \cap \EEE_2} \eqnl
  & \quad + \Pcond{\hrv d_N \le c}{\EEE_1\compl \cup \EEE_2\compl} \Prob{\EEE_1\compl \cup \EEE_2\compl} \\
                        &\ge \Pcond{\hrv d_N \le c}{\EEE_1 \cap \EEE_2} (1-\frac{\delta}{2}) \\
                        &= \Prob{\hat d_N(\rvt z_1, \rvt z_2) \le c} (1-\frac{\delta}{2}) \\
                        &\ge 1 - \delta .
\end{align}
In summary, for an estimator $\hat d_N$ and any $\delta, C > 0$ and $N \in \NN$, we can find distributions $\p{}$, $\q{}$ and $c \in \RR$, such that $\Prob[][normal]{\hrv d_N > c} \le \delta$, even though $D \ge c + C$.

\section{Discussion and Perspectives}
\label{sec:discussion}

We showed that under mild assumptions on the family of allowed distributions, differential entropy cannot be reliably estimated solely based on samples, no matter how many samples are available.
In particular, as first noted in \cite{Donoho1988One} no non-trivial bound or estimate of an information measure can be obtained based only on samples.
External information about the regularity of the underlying probability distribution needs to be taken into account. However, such regularity assumptions are not subject to empirical verification and thus, the existence of statistical guarantees for an empirical estimate cannot be empirically tested. This shows that researchers should take great care when approximating or bounding information measures, and specifically explore the necessary assumptions for the underlying distribution.

Regarding the use of information measures in machine learning, we note that our results apply to all estimators of information measures.
In particular, empirical versions of variational bounds cannot provide estimates of information measures with high reliability in general.

It would be interesting to investigate the type of assumptions on the underlying distributions that may hold in typical machine learning setups. However, as pointed out previously, these properties cannot be deduced from data, but must result from the model under consideration.
In a related note, it might be interesting if the confidence bounds for differential entropy estimation under bounded support and Lipschitz condition from \cref{sec:lipschitz} carry over to empirical versions of variational bounds.
Extensions of these results to other information measures, \eg, Rényi entropy, Rényi divergences, or $f$-divergences, could also be of particular interest for future work.

\acknowledgments

\begin{appendix}
  \section{Proof of Theorem \ref{thm:lipschitz}}
  \label{sec:proof:lipschitz}

  We shall first introduce auxiliary random variables which are depicted in \cref{fig:random-variables}. Let $\wrvt x$ be the element-wise discretization $\wrvt x = \Delta_M(\rvt x)$.
  The continuous random vector $\orvt x = \wrvt x + \uniform{[0,\frac 1M]^K}$ is obtained by adding independent uniform random noise. Let $\q{}$ be the \gls{pdf} of $\orvt x$.
  It is straightforward to see that $\wrvt y$ is distributed according to the empirical distribution of $N$ \iid copies of $\wrvt x$.
  Also note that $\ent[normal]{\wrvt x} = \dent[normal]{\orvt x} + K\log M$.
  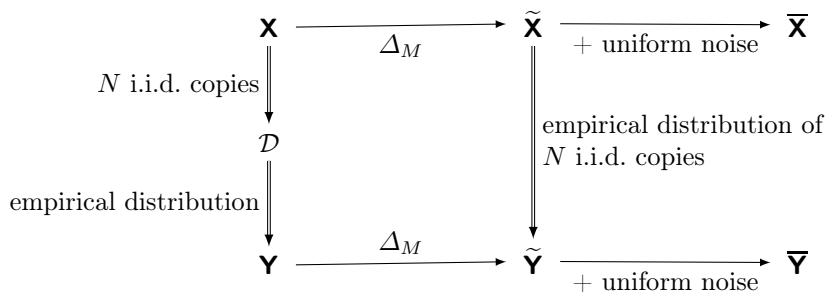
\begin{figure*}[tbh]
    \centering
    \begin{tikzpicture}
  \matrix (m) [matrix of math nodes,row sep=3em,column sep=8em,minimum width=2em]
  {
    \rvt x & \wrvt x & \orvt x \\
    \DDD \\
    \rvt y & \wrvt y & \orvt y \\
  };
  \path[-latex]
  (m-1-1) edge [double] node [left] {$N$ \iid copies} (m-2-1)
  edge node [below] {$\Delta_M$} (m-1-2)
  (m-1-2) edge node [below] {+ uniform noise} (m-1-3)
  (m-2-1) edge [double] node [left] {empirical distribution} (m-3-1)
  (m-3-1) edge node [above] {$\Delta_M$} node [below] {} (m-3-2)
  (m-3-2) edge node [below] {+ uniform noise} (m-3-3)
  (m-1-2) edge [double] node [right,text width=4cm] {empirical distribution of $N$ \iid copies} (m-3-2);
\end{tikzpicture}
    \caption{Connections between all involved random variables.}
    \label{fig:random-variables}
  \end{figure*}

  In order to prove \cref{thm:lipschitz} we use the triangle inequality twice to obtain
  \begin{align}
    & \big|\hat h(\DDD) - \dent{\rvt x} \big| \eqnl
    & \qquad = \left|\dent[normal]{\orvt y} - \dent{\rvt x} \right|  \\
    & \qquad \le \left|\dent[normal]{\orvt y} - \dent[normal]{\orvt x} \right| + \left| \dent[normal]{\orvt x} - \dent{\rvt x} \right| \\
    & \qquad = \big|\ent[normal]{\wrvt y} - \ent[normal]{\wrvt x} \big|
      + \left| \dent[normal]{\orvt x} - \dent{\rvt x} \right| \\
    & \qquad \le \big|\ent[normal]{\wrvt y} - \Exp[][normal]{\ent[normal]{\wrvt y}} \big|
      + \big|\Exp[][normal]{\ent[normal]{\wrvt y}} - \ent[normal]{\wrvt x} \big| \eqnl
    & \qquad \qquad + \left| \dent[normal]{\orvt x} - \dent{\rvt x} \right| , \label{eq:lipschitz:proof:2} 
  \end{align}
  noting that $\dent[normal]{\orvt y} - \ent[normal]{\wrvt y} = \dent[normal]{\orvt x} - \ent[normal]{\wrvt x} = -K\log M$. Note that $\ent[normal]{\wrvt y}$ is a random quantity that depends on $\DDD$.
  We thus split the bound in three terms, where the first term in \cref{eq:lipschitz:proof:2} is variance-like and the second and third terms constitute the bias.
  In the remainder of this \lcnamecref{sec:proof:lipschitz}, we will complete the proof by showing that all three terms in \cref{eq:lipschitz:proof:2} can be bounded as follows.
  First, we have
  \begin{align}
    \big|\Exp[][normal]{\ent[normal]{\wrvt y}} - \ent[normal]{\wrvt x} \big| &\le \log\left( 1 + \frac{M^K-1}{N} \right) , \text{ and} \label{eq:entropy_estimate:bias1}\\
    \left| \dent[normal]{\orvt x} - \dent{\rvt x} \right| &\le \frac{LK}{2M} \log(M \eta(K,L)). \label{eq:entropy_estimate:bias2}
  \end{align}
  And with probability greater than $1-\delta$ we also have
  \begin{align}
    \big|\ent[normal]{\wrvt y} - \Exp[][normal]{\ent[normal]{\wrvt y}} \big| &\le \sqrt{\frac{2}{N} \log\frac{2}{\delta}} \log N . \label{eq:entropy_estimate:var}
  \end{align}
  As $\wrvt y$ is distributed according to the empirical distribution of $N$ \iid copies of $\wrvt x$, on an alphabet of size $M^K$, the inequalities \cref{eq:entropy_estimate:bias1,eq:entropy_estimate:var} follow directly from the following well-known \lcnamecref{lem:entropy_estimate}, concerning the estimation of (discrete) Shannon entropy.   \begin{lemma}[{\cite[eq.~(3.4),~and~Prop.~1]{Paninski2003Estimation} and \cite[Remark~iii,~p.168]{Antos2001Convergence}}]
    \label{lem:entropy_estimate}
    Let $\rv z$ be a random variable on $\Nto{M}$ and $\hrv z$ distributed according to the empirical measure of $N$ \iid copies of $\rv z$. We then have $\big|\ent[normal]{\hrv z} - \ent[normal]{\rv z}\big| \le \big|\ent[normal]{\rv z} - \Exp[][normal]{\ent[normal]{\hrv z}}\big| + \big|\ent[normal]{\hrv z} - \Exp[][normal]{\ent[normal]{\hrv z}}\big|$, where
    \begin{align}
      \big|\ent[normal]{\rv z} - \Exp[][normal]{\ent[normal]{\hrv z}}\big| \le \log\left(1 + \frac{M-1}{N}\right),
    \end{align}
    and for any $\delta \in (0,1]$, with probability greater than $1-\delta$,
    \begin{align}
      \big|\ent[normal]{\hrv z} - \Exp[][normal]{\ent[normal]{\hrv z}}\big| \le \sqrt{\frac{2}{N}\log\frac{2}{\delta}} \log N .
    \end{align}
  \end{lemma}
                
  In order to show \cref{eq:entropy_estimate:bias2}, we will first obtain some preliminary results and then conclude the proof in \cref{lem:discretizaion_bound}. We start by bounding the difference between $\p{}$ and the approximation $\q{}$ using the following auxiliary results.
  \begin{lemma}
    \label{lem:f_difference_bound}
    Let $f\colon [\vt 0, \vt A] \to \RR$ be an arbitrary $L$-Lipschitz continuous function and assume $f(\vt y) = 0$ for some\footnote{We use the notation $[\vt 0, \vt A] = \{\vt x \in \RR^K: 0 \le x_k \le A_k, k \in \Nto{K}\}$.} $\vt y \in [\vt 0, \vt A]$, then
    \begin{align}
      \int_{[\vt 0,\vt A]} |f(\vt x)| \,\lebesgue^K(\dd \vt x) \le \frac L2 \left(\prod_{k=1}^K A_{k} \right) \left(\sum_{k=1}^K  A_k \right) .
    \end{align}
    In particular, for $A_k = \eps$, $k \in \Nto{K}$, we have $ \int_{[0,\eps]^K} |f(\vt x)| \,\lebesgue^K(\dd \vt x) \le \eps^{K+1} \frac{LK}{2}$.
  \end{lemma}
  \begin{proof}
    For $\vt x \in [\vt 0, \vt A]$ we have $|f(\vt x)| = |f(\vt x) - f(\vt y)| \le L\norm{\vt x - \vt y}$ and hence 
    \begin{align}
      & \int_{[\vt 0,\vt A]} |f(\vt x)| \,\lebesgue^K(\dd \vt x) \eqnl
      & \qquad \le L\int_{[\vt 0,\vt A]} \norm{\vt x - \vt y} \,\lebesgue^K(\dd \vt x) \eqnl
      & \qquad = L \sum_{k=1}^K \int_{[\vt 0,\vt A]} |x_k - y_k| \,\lebesgue^K(\dd \vt x) \\
                                          & \qquad = \frac L2 \left(\prod_{k=1}^K A_{k} \right) \left(\sum_{k=1}^K  \frac{A_k^2 - 2 y_k (A_k - y_k)}{A_k} \right) \\
      & \qquad \le \frac L2 \left(\prod_{k=1}^K A_{k} \right) \left(\sum_{k=1}^K  A_k \right) . \qedhere
    \end{align}
  \end{proof}
  
  \begin{lemma}
    \label{lem:difference_bound}
    For an $L$-Lipschitz continuous \gls{pdf} $\p{}$ on $\XXX$ and $\q{}$ the \gls{pdf} of $\orvt x$, as defined above, we have $|\p{\vt x} - \q{\vt x}| \le \frac{LK}{2M}$ for every $\vt x \in \XXX$.
  \end{lemma}
  \begin{proof}
    Let $\vt x \in \XXX$ and $\wvt x = \Delta_M(\vt x)$. The function $\q{}$ is constant on $\Delta_M^{-1}(\wvt x)$ and given by $\q{\vt x'} = \lebesgue^K(\Delta_M^{-1}(\wvt x))^{-1} \int_{\Delta_M^{-1}(\wvt x)} \p{} \,\dd\lebesgue^K$ for all $\vt x' \in \Delta_M^{-1}(\wvt x)$, where $\lebesgue^K(\Delta_M^{-1}(\wvt x)) = M^{-K}$. Thus, since $\vt x \in \Delta_M^{-1}(\wvt x)$, we obtain
    \begin{align}
      M^{-K} |\q{\vt x} - \p{\vt x}| &= \left| \int_{\Delta_M^{-1}(\wvt x)} \p{\vt x'} - \p{\vt x} \,\lebesgue^K(\dd \vt x') \right| \\
                                     &\le \int_{\Delta_M^{-1}(\wvt x)} |\p{\vt x'} - \p{\vt x}| \,\lebesgue^K(\dd \vt x') \\
                                     &\le M^{-K-1} \frac{LK}{2} \label{eq:difference_bound_apply}, 
    \end{align}
    where we applied \cref{lem:f_difference_bound} to $\vt x' \mapsto \p{\vt x'} - \p{\vt x}$ in \cref{eq:difference_bound_apply}.
  \end{proof}

  \begin{lemma}
    \label{lem:upper_bound}
    If $\p{}$ is an $L$-Lipschitz continuous \gls{pdf} on $\RR^K$, then $\p{} \le \left(\frac{L^K (K+1)!}{2^K}\right)^{\frac{1}{K+1}}$.
  \end{lemma}
  \begin{proof}
    Let $\vt x \in \RR^K$ and define the ball $\Ball_{\vt x}(\frac{\p{\vt x}}{L}) = \{\vt x' \in \RR^K : \norm{\vt x' - \vt x} \le \frac{\p{\vt x}}{L}\}$ with radius $\frac{\p{\vt x}}{L}$, centered at $\vt x$. We then have $\lebesgue^K(\Ball_{\vt x}(\frac{\p{\vt x}}{L})) = \frac{2^K \p{\vt x}^{K}}{L^{K} K!}$ and hence,
    \begin{align}
      1 &\ge \int_{\Ball_{\vt x}(\frac{\p{\vt x}}{L})} \p{\vt x} \,\lebesgue^K(\dd \vt x ) \\
        &= \frac{2^K \p{\vt x}^{K+1}}{L^{K} K!} - \int_{\Ball_{\vt x}(\frac{\p{\vt x}}{L})} \p{\vt x} - \p{\vt x'} \,\lebesgue^K(\dd \vt x') \\
        &\ge \frac{2^K \p{\vt x}^{K+1}}{L^{K} K!} - L \int_{\Ball_{\vt x}(\frac{\p{\vt x}}{L})} \norm{\vt x - \vt x'} \,\lebesgue^K(\dd\vt x')  \label{eq:upper_bound_use_Lipschitz} \\
                                        &= \frac{2^K \p{\vt x}^{K+1}}{L^{K}} \frac{1}{(K+1)!} ,
    \end{align}
    where the fact that $\p{\vt x} - \p{\vt x'} \le L \norm{\vt x - \vt x'}$ for all $\vt x' \in \RR^K$ is used in~\cref{eq:upper_bound_use_Lipschitz}.
  \end{proof}

  Using the previous lemmas to bound the distance between $\p{}$ and $\q{}$, the following two results will allow us to bound the difference of the differential entropies.
  \begin{lemma}
    \label{lem:diff_info_bound}
    For $x \in [0,1]$, $y \ge 0$, and $a \defas |x-y| \le \alpha \approx \pgfmathprintnumber[fixed,precision=3]{\myalpha}$ we have
    \begin{align}
      \label{eq:diff_info_bound}
      |x\log x - y \log y| \le - a \log a .    
    \end{align}
  \end{lemma}
  \begin{proof}
    In the following, we assume w.l.o.g.\ that $x \le y = x + a$.
    If $x \le y \le \ee{-1}$, then $|x \log x - (x + a) \log (x + a)| = x \log x - (x + a) \log (x + a)$ is monotonically decreasing in $x$ and thus maximal at $x=0$ and hence, \cref{eq:diff_info_bound} follows.
    If, on the other hand, $y \ge \ee{-1}$, then necessarily $\alpha \le \ee{-1} - \alpha \le x \le y \le 1+\alpha$.
    Define the function $f(x) \defas -x\log x$, $x > 0$ and $f(0) \defas 0$. Note that by the mean value theorem there are $a_0 \in (0,a)$ and $x_0 \in (x,y)$ such that $|f(0) - f(a)| = f(a) = a f'(a_0)$ and $|f(x) - f(y)| = a |f'(x_0)|$. Inequality \cref{eq:diff_info_bound} then follows by observing that $|f'(x_0)| \le f'(a_0)$ whenever $a_0 \in (0,\alpha)$ and $x_0 \in (\alpha, 1+\alpha)$.
  \end{proof}

  \begin{lemma}
    \label{lem:discretizaion_bound1}
    Let $\p{}$ and $\q{}$ be two \glspl{pdf} supported on $\XXX$ with finite differential entropies. Assume that for all $\vt x \in \XXX$ we have $|\p{\vt x} - \q{\vt x}| \le \eps$ and $0 \le \p{\vt x} \le A$, and that $\frac{\eps}{A} \le \alpha$ holds. Then,
    \begin{align}
      |\dent{\p{}} - \dent{\q{}}| \le \eps \log\frac{A}{\eps} . \label{eq:discretizaion_bound1}
    \end{align}
  \end{lemma}
  \begin{proof}
    Define $\pp{\vt x} \defas A^{-1} \p{A^{-\frac{1}{K}} \vt x}$ 
        and
    $\qq{\vt x} \defas A^{-1} \q{A^{-\frac{1}{K}} \vt x}$ for $\vt x \in [0,A^{\frac 1K}]^K$. We have $|\pp{\vt x} - \qq{\vt x}| \le A^{-1}\eps$, $0 \le \pp{\vt x} \le 1$ as well as
    \begin{align}
      & |\dent{\p{}} - \dent{\q{}}| \eqnl
      &\qquad = |\dent{\pp{}} - \dent{\qq{}}| \\
      &\qquad = \left| \int \pp{} \log \pp{} - \qq{} \log\qq{} \,\dd\lebesgue^K \right| \\
      &\qquad \le \int \left| \pp{} \log \pp{} - \qq{} \log\qq{} \right| \,\dd\lebesgue^K \\
      &\qquad \le - \lebesgue^K([0,A^{\frac 1K}]^K) \cdot A^{-1} \eps \log(A^{-1}\eps) \label{eq:discretization_apply_diffbound} \\
            &\qquad= \eps \log\frac{A}{\eps} ,
    \end{align}
    where \cref{lem:diff_info_bound} was applied in \cref{eq:discretization_apply_diffbound}.
  \end{proof}

  We can now finish the proof of \cref{thm:lipschitz} by showing \cref{eq:entropy_estimate:bias2}.
  \begin{lemma}
    \label{lem:discretizaion_bound}
    If $M \ge \frac{1}{\alpha \eta(K,L)}$ we have
    \begin{align}
      \left| \dent[normal]{\orvt x} - \dent{\rvt x} \right| \le \frac{LK}{2M} \log(M \eta(K,L)) .
    \end{align}
  \end{lemma}
  \begin{proof}
    By \cref{lem:difference_bound}, $|\p{\vt x} - \q{\vt x}| \le \frac{LK}{2M}$ and, by \cref{lem:upper_bound}, $\p{} \le \left(\frac{L^K (K+1)!}{2^K}\right)^{\frac{1}{K+1}}$. We can thus apply \cref{lem:discretizaion_bound1} with $\eps = \frac{LK}{2M}$ and $A = \left(\frac{L^K (K+1)!}{2^K}\right)^{\frac{1}{K+1}}$ provided that $\frac{\eps}{A} \le \alpha$, which is equivalent to $M \ge \frac{1}{\alpha \eta(K,L)}$. Inserting $\eps$ and $A$ in~\cref{eq:discretizaion_bound1} proves the result.
  \end{proof}

 \end{appendix}

\newpage

\bibliography{IEEEabrv.bib,literature}

\begin{thebibliography}{10}

\bibitem{Ahmad1976Nonparametric}
I.~Ahmad and P.-E. Lin.
\newblock A nonparametric estimation of the entropy for absolutely continuous
  distributions.
\newblock {\em {IEEE} Trans. Inf. Theory}, 22(3):372--375, May 1976.

\bibitem{Amjad2019Learning}
R.~A. {Amjad} and B.~C. {Geiger}.
\newblock Learning representations for neural network-based classification
  using the information bottleneck principle.
\newblock {\em IEEE Transactions on Pattern Analysis and Machine Intelligence},
  2019.
\newblock to appear.

\bibitem{Antos2001Convergence}
A.~Antos and I.~Kontoyiannis.
\newblock Convergence properties of functional estimates for discrete
  distributions.
\newblock {\em Random Structures \& Algorithms}, 19(3-4):163--193, Nov. 2001.

\bibitem{Bahadur1956Nonexistence}
R.~R. Bahadur and L.~J. Savage.
\newblock The nonexistence of certain statistical procedures in nonparametric
  problems.
\newblock {\em Ann. Math. Statist.}, 27(4):1115--1122, 1956.

\bibitem{Barber2003IM}
D.~Barber and F.~Agakov.
\newblock The {IM} algorithm: A variational approach to information
  maximization.
\newblock In {\em NIPS'03}, volume~16 of {\em Advances in neural information
  processing systems}, pages 201--208, Cambridge, MA, USA, Dec. 2003.

\bibitem{Beirlant1997Nonparametric}
J.~Beirlant, E.~J. Dudewicz, L.~Gy{\"o}rfi, and E.~C. Van~der Meulen.
\newblock Nonparametric entropy estimation: An overview.
\newblock {\em International Journal of Mathematical and Statistical Sciences},
  6(1):17--39, 1997.

\bibitem{Belghazi2018Mutual}
M.~I. Belghazi, A.~Baratin, S.~Rajeshwar, S.~Ozair, Y.~Bengio, A.~Courville,
  and D.~Hjelm.
\newblock Mutual information neural estimation.
\newblock In {\em ICML'18}, volume~80 of {\em PMLR}, pages 531--540, Stockholm,
  Sweden, July 2018.

\bibitem{Chen2016InfoGAN}
X.~Chen, Y.~Duan, R.~Houthooft, J.~Schulman, I.~Sutskever, and P.~Abbeel.
\newblock Info{GAN}: Interpretable representation learning by information
  maximizing generative adversarial nets.
\newblock In {\em NIPS'16}, volume~29 of {\em Advances in neural information
  processing systems}, pages 2172--2180, Barcelona, Spain, 2016.

\bibitem{Darbellay1999Estimation}
G.~A. Darbellay and I.~Vajda.
\newblock Estimation of the information by an adaptive partitioning of the
  observation space.
\newblock {\em {IEEE} Trans. Inf. Theory}, 45(4):1315--1321, May 1999.

\bibitem{Donoho1988One}
D.~L. Donoho.
\newblock One-sided inference about functionals of a density.
\newblock {\em Ann. Statist.}, 16(4):1390--1420, 1988.

\bibitem{Donoho1991Geometrizing}
D.~L. Donoho and R.~C. Liu.
\newblock Geometrizing rates of convergence, {II}.
\newblock {\em Ann. Statist.}, 19(2):633--667, 1991.

\bibitem{Gao2017Estimating}
W.~Gao, S.~Kannan, S.~Oh, and P.~Viswanath.
\newblock Estimating mutual information for discrete-continuous mixtures.
\newblock In {\em NIPS'17}, volume~30 of {\em Advances in Neural Information
  Processing Systems}, pages 5986--5997, Long Beach, CA, USA, 2017.

\bibitem{Gleser1987Nonexistence}
L.~J. Gleser and J.~T. Hwang.
\newblock The nonexistence of 100 (1-$\alpha$)\% confidence sets of finite
  expected diameter in errors-in-variables and related models.
\newblock {\em Ann. Statist.}, pages 1351--1362, 1987.

\bibitem{Goldfeld2019Convergence}
Z.~Goldfeld, K.~Greenewald, Y.~Polyanskiy, and J.~Weed.
\newblock Convergence of smoothed empirical measures with applications to
  entropy estimation.
\newblock {\em arXiv preprint}, 2019.

\bibitem{Goldfeld2019Estimating}
Z.~Goldfeld, E.~Van Den~Berg, K.~Greenewald, I.~Melnyk, N.~Nguyen,
  B.~Kingsbury, and Y.~Polyanskiy.
\newblock Estimating information flow in deep neural networks.
\newblock In {\em ICML'19}, volume~97 of {\em PMLR}, pages 2299--2308, Long
  Beach, CA, USA, 2019.

\bibitem{Gordon2003Applying}
S.~Gordon, H.~Greenspan, and J.~Goldberger.
\newblock Applying the information bottleneck principle to unsupervised
  clustering of discrete and continuous image representations.
\newblock In {\em Proc.\ Ninth IEEE Int. Conf. Comput. Vision}, pages 370--377,
  Nice, France, Oct. 2003.

\bibitem{Gyoerfi1987Density}
L.~Gy{\"o}rfi and E.~C. Van~der Meulen.
\newblock Density-free convergence properties of various estimators of entropy.
\newblock {\em Comput. Stat. Data Anal.}, 5(4):425--436, Sept. 1987.

\bibitem{Han2017Optimal}
Y.~Han, J.~Jiao, T.~Weissman, and Y.~Wu.
\newblock Optimal rates of entropy estimation over {L}ipschitz balls.
\newblock {\em arXiv preprint}, 2019.

\bibitem{Hjelm2018Learning}
R.~D. Hjelm, A.~Fedorov, S.~Lavoie-Marchildon, K.~Grewal, P.~Bachman,
  A.~Trischler, and Y.~Bengio.
\newblock Learning deep representations by mutual information estimation and
  maximization.
\newblock In {\em ICLR}, New Orleans, LA, USA, 2019.

\bibitem{Hu2017Learning}
W.~Hu, T.~Miyato, S.~Tokui, E.~Matsumoto, and M.~Sugiyama.
\newblock Learning discrete representations via information maximizing
  self-augmented training.
\newblock In {\em ICML'17}, volume~70 of {\em PMLR}, pages 1558--1567, Sydney,
  Australia, 2017.

\bibitem{Kandasamy2015Nonparametric}
K.~Kandasamy, A.~Krishnamurthy, B.~Poczos, L.~Wasserman, and J.~M. Robins.
\newblock Nonparametric von {M}ises estimators for entropies, divergences and
  mutual informations.
\newblock In {\em NIPS'15}, volume~28 of {\em Advances in Neural Information
  Processing Systems}, pages 397--405, Montr\'eal, Canada, 2015.

\bibitem{Kolchinsky2018Caveats}
A.~Kolchinsky, B.~D. Tracey, and S.~V. Kuyk.
\newblock Caveats for information bottleneck in deterministic scenarios.
\newblock In {\em ICLR}, New Orleans, LA, USA, 2019.

\bibitem{Liu2012Exponential}
H.~Liu, L.~Wasserman, and J.~D. Lafferty.
\newblock Exponential concentration for mutual information estimation with
  application to forests.
\newblock In {\em NIPS'12}, volume~26 of {\em Advances in Neural Information
  Processing Systems}, pages 2537--2545, Lake Tahoe, NV, USA, 2012.

\bibitem{McAllester2018Formal}
D.~McAllester and K.~Stratos.
\newblock Formal limitations on the measurement of mutual information.
\newblock In {\em ICLR}, New Orleans, LA, USA, 2019.

\bibitem{Miyato2015Distributional}
T.~Miyato, S.-i. Maeda, M.~Koyama, K.~Nakae, and S.~Ishii.
\newblock Distributional smoothing with virtual adversarial training.
\newblock In {\em ICLR}, San Juan, Puerto Rico, 2016.

\bibitem{Nemenman2004Entropy}
I.~Nemenman, W.~Bialek, and R.~D.~R. Van~Steveninck.
\newblock Entropy and information in neural spike trains: Progress on the
  sampling problem.
\newblock {\em Phys. Rev. E}, 69(5):056111, 2004.

\bibitem{Nguyen2010Estimating}
X.~Nguyen, M.~J. Wainwright, and M.~I. Jordan.
\newblock Estimating divergence functionals and the likelihood ratio by convex
  risk minimization.
\newblock {\em {IEEE} Trans. Inf. Theory}, 56(11):5847--5861, Nov. 2010.

\bibitem{Paninski2003Estimation}
L.~Paninski.
\newblock Estimation of entropy and mutual information.
\newblock {\em Neural Comput.}, 15(6):1191--1253, 2003.

\bibitem{Pfanzagl1998Nonexistence}
J.~Pfanzagl.
\newblock The nonexistence of confidence sets for discontinuous functionals.
\newblock {\em J. Stat. Plan. Inference}, 75(1):9--20, 1998.

\bibitem{Poole2019Variational}
B.~Poole, S.~Ozair, A.~van~den Oord, A.~A. Alemi, and G.~Tucker.
\newblock On variational bounds of mutual information.
\newblock In {\em ICML'19}, volume~97 of {\em PMLR}, pages 5171--5180, Long
  Beach, CA, USA, 2019.

\bibitem{Shannon1948Mathematical}
C.~E. Shannon.
\newblock A mathematical theory of communication.
\newblock {\em Bell Syst. Tech. J.}, 27(3):379--423, July 1948.

\bibitem{Shannon1951Prediction}
C.~E. Shannon.
\newblock Prediction and entropy of printed {E}nglish.
\newblock {\em Bell Ssyst. Tech. J.}, 30(1):50--64, 1951.

\bibitem{Shwartz-Ziv2017Opening}
R.~Shwartz-Ziv and N.~Tishby.
\newblock Opening the black box of deep neural networks via information.
\newblock {\em arXiv preprint}, 2017.

\bibitem{Singh2014Generalized}
S.~Singh and B.~Poczos.
\newblock Generalized exponential concentration inequality for renyi divergence
  estimation.
\newblock In {\em ICML'14}, volume~32 of {\em PMLR}, pages 333--341, Bejing,
  China, 2014.

\bibitem{Singh2016Finite}
S.~Singh and B.~Poczos.
\newblock Finite-sample analysis of fixed-k nearest neighbor density functional
  estimators.
\newblock In {\em NIPS'16}, volume~29 of {\em Advances in Neural Information
  Processing Systems}, pages 1217--1225, Barcelona, Spain, 2016.

\bibitem{Sricharan2011k}
K.~Sricharan, R.~Raich, and A.~O. Hero.
\newblock k-nearest neighbor estimation of entropies with confidence.
\newblock In {\em Proc. IEEE Int. Symp. Inf. Theory (ISIT 2011)}, pages
  1205--1209, Saint Petersburg, Russia, July 2011.

\bibitem{Tishby2000Information}
N.~Tishby, F.~C. Pereira, and W.~Bialek.
\newblock The information bottleneck method.
\newblock In {\em Annu. Allerton Conf. Commun., Control, and Comput.}, pages
  368--377, Monticello, IL, Sept. 1999.

\bibitem{Oord2018Representation}
A.~van~den Oord, Y.~Li, and O.~Vinyals.
\newblock Representation learning with contrastive predictive coding.
\newblock {\em arXiv preprint}, 2018.

\bibitem{Wang2005Divergence}
Q.~Wang, S.~R. Kulkarni, and S.~Verd{\'{u}}.
\newblock Divergence estimation of continuous distributions based on
  data-dependent partitions.
\newblock {\em {IEEE} Trans. Inf. Theory}, 51(9):3064--3074, Sept. 2005.

\end{thebibliography}
\bibliographystyle{abbrv}

\end{document}